\documentclass{article}
\usepackage{fullpage} 

\usepackage{mdwlist}

\usepackage{amssymb,amsbsy,latexsym,amsthm}
\usepackage{amsmath}
\usepackage{graphics, subfigure, float}
\usepackage{fp, calc}
\usepackage{bm}
\usepackage{hyperref}
\usepackage[latin1]{}
\usepackage[english]{babel} 
\usepackage[T1]{fontenc} 

\usepackage{bm}

\usepackage{datetime}
\usepackage{verbatim, comment}

\usepackage{pst-all}

\newtheoremstyle{theorem}{1em}{1em}{\slshape}{0pt}{\bfseries}{.}{ }{}
\theoremstyle{theorem}
\newtheorem{theorem}{Theorem}
\newtheorem*{theorem*}{Theorem}

\newtheorem{lemma}[theorem]{Lemma}

\theoremstyle{remark}

\newtheorem*{remark*}{Remark}

\providecommand{\setN}{\mathbb{N}}
\providecommand{\setZ}{\mathbb{Z}}

\providecommand{\setR}{\mathbb{R}}

\theoremstyle{theorem}
\theoremstyle{definition}
\newtheorem{definition}{Definition}
\newtheorem*{claim*}{Claim}

\usepackage[displaymath,textmath,graphics, subfigure, floats]{preview} 
\PreviewEnvironment{enumerate} 
\PreviewEnvironment{itemize}
\PreviewEnvironment{theorem} 
\PreviewEnvironment{lemma} 
\PreviewEnvironment{myproblem}
\PreviewEnvironment{algorithm} 
\PreviewEnvironment{defn} 
\PreviewEnvironment{center} 
\PreviewEnvironment{pspicture} 
\usepackage{datetime}

\makeatother

\usepackage{calrsfs} 
\DeclareMathAlphabet{\pazocal}{OMS}{zplm}{m}{n}

\begin{document}

\title{Number Balancing is as hard as Minkowski's Theorem and Shortest Vector}
\date{University of Washington, Seattle, WA 98195} 
\author{Rebecca Hoberg\thanks{Email: {\tt rahoberg@uw.edu}} \and Harishchandra Ramadas\thanks{Email: {\tt ramadas@uw.edu}} \and Thomas Rothvoss\thanks{Email: {\tt rothvoss@uw.edu}. Supported by NSF grant 1420180 with title ``\emph{Limitations of convex relaxations in combinatorial optimization}'',  an Alfred P. Sloan Research Fellowship and a David \& Lucile Packard Foundation Fellowship.} \and Xin Yang\thanks{Email: {\tt yx1992@uw.edu}}} 

\maketitle

\begin{abstract}
The number balancing (NBP) problem is the following: given real numbers
$a_1,\ldots,a_n \in [0,1]$, find two disjoint subsets $I_1,I_2 \subseteq [n]$ so 
that the difference $|\sum_{i \in I_1}a_i - \sum_{i \in I_2}a_i|$ of their sums is minimized. An application of 
the pigeonhole principle shows that there is always a solution where
the difference is at most $O(\frac{\sqrt{n}}{2^n})$. 
Finding the minimum, however, is NP-hard. 
In polynomial time,
the \emph{differencing algorithm} by Karmarkar and Karp from 1982 can produce a solution
with difference at most $n^{-\Theta(\log n)}$, but no further
improvement has been made since then. 

In this paper, we show a relationship between NBP and Minkowski's Theorem. First we show that an approximate oracle for Minkowski's Theorem gives an approximate NBP oracle. Perhaps more surprisingly, we show that an approximate NBP oracle gives an approximate Minkowski oracle. In particular, we prove that any polynomial time algorithm that 
guarantees a solution of difference at most $2^{\sqrt{n}} / 2^{n}$ would give a polynomial approximation for Minkowski as well as a polynomial factor approximation algorithm for the Shortest Vector Problem.
\end{abstract}

\section{Introduction}

One of \emph{six basic NP-complete problems} of Garey and Johnson~\cite{GJNPCompleteness} is the 
\emph{partition problem} that for a list of numbers $a_1,\ldots,a_n$ asks whether there is a partition of the indices 
so that the sums of the numbers in both partitions coincide. 
 Partition and related problems like knapsack, subset sum and 
bin packing are some of the fundamental classical problems in theoretical computer science with numerous practical applications; see for example the textbooks \cite{KnapsackProblemsMartelloTothBook1990,KnapsackProblemsKellererEtAlBook2004}
and the article of Mertens~\cite{mertens2006easiest}. 
In this paper, we study a variant called the \emph{number balancing problem} (NBP), where the goal is to find two disjoint subsets $I_1,I_2 \subseteq \{ 1,\ldots,n\}$ so that the difference 
$|\sum_{i \in I_1} a_i - \sum_{i \in I_2} a_i|$ is minimized. 
Equivalently, given a vector of numbers $\bm a = (a_1,\ldots,a_n) \in [0,1]^n$, we want to find a vector of \emph{signs} $\bm x \in \lbrace{-1,0,1\rbrace}^n\setminus\{\bm{0}\}$ so that $\left|\left<\bm a,\bm x\right>\right| = \left|\sum_{i=1}^n x_i a_i\right|$ is minimized.
 Woeginger and Yu~\cite{woeginger1992equal} 
studied this problem under the name ``equal-subset-sum'' 
and showed that it is NP-hard to decide whether there are two disjoint subsets that sum up to the 
exact same value. This version has also been extensively studied in combinatorics~\cite{lunnon1988integer,bohman1996sum,lev2011size}.

On the positive side, it is not hard to prove that there is always a solution with exponentially small error.
Suppose that $a_1,\ldots,a_n \in [0,1]$.
Consider the list of $2^n$ many numbers $\sum_{i=1}^n a_ix_i$ for all $\bm{x} \in \{ 0,1\}^n$. All these numbers fall into the interval $[0,n]$, hence by the \emph{pigeonhole principle}, we can find two distinct vectors $\bm{x},\bm{x}' \in \{ 0,1\}^n$ with $|\sum_{i=1}^n a_ix_i - \sum_{i=1}^n a_ix_i'| \leq \frac{n}{2^n-1}$. Then $\bm{x}-\bm{x}'$ gives the desired solution. Note that the bound can be slightly improved to $O(\frac{\sqrt{n}}{2^n})$ by using the fact that due 
to \emph{concentration of measure} effects, for a constant
fraction of vectors $\bm{x} \in \{ 0,1\}^n$, the sums $\sum_{i=1}^n a_ix_i$ fall into an interval of length $\sqrt{n}$ (instead of $n$). 

However, since these arguments rely on the pigeonhole principle, they are non-constructive. Restricting the non-constructive argument to polynomially many ``pigeons'' provides a 
simple polynomial time algorithm to find at least an $\bm{x} \in \{ -1,0,1\}^n \setminus \{\bm{0}\}$ with $|\left<\bm{a},\bm{x}\right>| \leq \frac{1}{\textrm{poly}(n)}$ for an arbitrarily small polynomial. Interestingly, the only known
polynomial time algorithm that gives a better guarantee is Karmarkar and Karp's \emph{differencing algorithm}~\cite{KKDifferencing}
which provides the bound $|\left<\bm{a},\bm{x}\right>| \leq n^{-c\log(n)}$ for some constant $c>0$. 
Their algorithm uses a recursive scheme; find $\Theta(n)$ pairs of numbers $a_i$ of distance at most $\Theta(\frac{1}{n})$ and create an instance consisting of their differences, then recurse.

This leads to the natural question: \emph{Given $a_1,\ldots,a_n \in [0,1]$, what upper bound on $|\sum_{i=1}^n a_ix_i|$ can be guaranteed if $\bm{x} \in \{ -1,0,1\}^n \setminus \{\bm{0}\}$ is to be chosen in polynomial time?}

While answering this question directly seems out of reach, it appears that NBP falls into a class of problems where good solutions exist due to the pigeonhole principle, 
such as the \emph{Shortest Vector Problem} or \emph{Minkowski's Theorem}. Recall that given linearly independent 
vectors $\bm{b}_1,\ldots,\bm{b}_n \in \setR^n$, a \emph{(full rank) lattice} is the set $\Lambda := \{ \sum_{i=1}^n \lambda_i\bm{b}_i : \lambda_i \in \setZ \; \forall i=1,\ldots,n\}$. 
The set $\{\bm{b}_1,\ldots,\bm{b}_n\}$ is called a $\emph{basis}$ for $\Lambda$ and we define $\det(\Lambda):=|\det(B)|$.
The Shortest vector problem then consists of finding a non-zero vector in $\Lambda$ that minimizes the Euclidean norm. 
The famous LLL-algorithm~\cite{lenstra1982factoring} can find a $2^{n/2}$-approximation in polynomial time (the generalized \emph{block reduction method} of
Schnorr~\cite{BlockReductionSchnorr87} brings the factor down to $2^{n \log \log(n) / \log(n)}$). As a rarity in theoretical computer science, the shortest 
vector problem admits $(\textrm{NP} \cap \textrm{coNP})$-certificates for a value that is at most a factor $O(\sqrt{n})$
away from the optimum~\cite{LatticeProblemsInNPintersec-coNP-AharonovRegevJACM05}, while the best known hardness lies at a subpolynomial bound of $n^{\Theta(1/\log \log n)}$~\cite{SVPhardness-RegevHavivSTOC07}.
Using the pigeonhole principle one can show that a lattice $\Lambda$ contains a vector of length at most 
$O(\sqrt{n}) \cdot \det(\Lambda)^{1/n}$. Interestingly, a polynomial time algorithm that achieves this bound would imply an
$O(n)$-approximation algorithm even for worst-case lattices~\cite{ajtai1996generating},  
enough to break lattice-based cryptosystems~\cite{lagarias1990korkin}. 

\emph{Minkowski's Theorem} tells us that any symmetric convex body $K \subseteq \setR^n$ of volume at least $2^n$ must intersect $\setZ^n \setminus \{\bm{0}\}$, see for example~\cite{LecturesOnDiscreteGeometryMatousek2002}. 
This theorem is proven by placing translates of $\frac{1}{2}K$ at any lattice point and then inferring an 
overlap due to the pigeonhole principle. Again, one can consider the algorithmic question: \emph{given a symmetric convex body $K$ with volume at least $2^n$, for what factor $\rho$ can one be guaranteed to find an $\bm{x} \in (\rho K) \cap (\setZ^n \setminus \{ \bm{0}\})$ in polynomial time?} We would like to point out that 
this factor $\rho$ is within a polynomial factor of the shortest vector approximability using the fact that there is a linear transformation sandwiching $K$ between two Euclidean balls whose radius differs by a factor of $\sqrt{n}$~\cite{JohnsTheorem1948}. 

It is not hard to use an \emph{exact} oracle for Minkowski's Theorem to find a good number balancing solution, since the body $K := \{ \bm{x} \in \left(-2,2\right) : |\sum_{i=1}^n x_ia_i| \leq \Theta(\frac{n}{2^n})\}$ has a volume of $2^n$.
However, it is not clear how we could use an \emph{approximate} oracle.
For example, it is known that the LLL-algorithm can be used to find 
a nonzero integer vector $\bm{x} \in \rho K$ for a factor of $\rho = \textrm{poly}(n) \cdot 2^{n/2}$. 
While the error guarantee of $|\sum_{i=1}^n a_ix_i| \leq \textrm{poly}(n) \cdot 2^{-n/2}$ outperforms the 
Karmarkar-Karp algorithm, we
only know that $\|\bm{x}\|_{\infty} < 2\rho$, which means that $\bm{x}$ will not be a valid solution if $\rho>1$\footnote{If $\rho\leq 2-\varepsilon$, then one can still obtain an error of $|\sum_{i=1}^n a_ix_i| \leq 2^{-\Theta(\varepsilon n)}$, but this breaks down if $\rho \geq 2$.}. This leads us to the next question: \emph{what factor $\rho$ is needed for Minkowski's Theorem to improve over Karmarkar-Karp's bound}? Again, we should mention that a $\textrm{poly}(n)$-approximation for Minkowski is not known
to be inconsistent with $\mathrm{NP} \neq \mathrm{P}$.

We have seen that in a certain sense the Shortest Vector Problem and Minkowski's Theorem
are generalizations of number balancing. This brings us to the question about the reverse: \emph{given an oracle that solves NBP within an exponentially small error, can this give a non-trivial oracle for the Shortest Vector Problem or Minkowski's Theorem?}

\subsection{Contribution}

In this work, we provide some answers to the questions raised above, by relating the complexity of the number balancing
problem to Minkowski's Theorem. For $\rho \geq 1$, we define a \emph{$\rho$-approximation for the Minkowski problem} as a
polynomial time algorithm that on input\footnote{We will specify in the text, how the body $K$ is represented.} of a symmetric convex set $K \subseteq \setR^n$
with $\textrm{vol}_n(K) > 2^n$, finds a vector in $(\rho \cdot K) \cap (\setZ^n \setminus \{\bm{0}\})$. Moreover, for $\delta>0$ we define a \emph{$\delta$-approximation}
for the number partitioning problem as a polynomial time algorithm that receives $\bm{a} \in [0,1]^n$ 
as input and produces a vector $\bm{x} \in \{ -1,0,1\}^n \setminus\{\bm{0}\}$ with $|\left<\bm{a},\bm{x}\right>| \leq \delta$. We provide the following reduction:  
\begin{theorem}\label{thm:minktonpp}
Suppose there is a $\rho$-approximation for Minkowski's problem for polytopes $K$ with $O(n)$ facets. 
Then there is a $\delta$-approximation for number balancing where $\delta := 2^{-n^{\Theta(1/\rho)}}$
\end{theorem}
In fact, it suffices to have such an oracle for the linear transformation of a cube, which is equivalent to an oracle for SVP in the $\|\cdot \|_\infty$ norm. 
\begin{theorem} \label{thm:ReducingNBPToSVP}
Suppose that there is a polynomial time algorithm that given a lattice 
$\Lambda \subseteq \setR^n$ with $\det(\Lambda) \leq 1$ finds a non-zero vector $\bm{x} \in \Lambda$ of length $\|\bm{x}\|_{\infty} \leq \rho$.
Then there is $\delta$-approximation for number balancing where $\delta := 2^{-n^{\Theta(1/\rho)}}$.
\end{theorem}
Recall that $\textrm{vol}_n([-1,1]^n) = 2^n$ and hence for such a lattice there is always an $\bm{x} \in \Lambda \setminus \{\bm{0}\}$ with
$\|\bm{x}\|_{\infty} \leq 1$. 
In particular an oracle for $\rho \leq c'\log(n)/\log \log (n)$ would imply an improvement over Karmarkar-Karp's
algorithm, where $c'>0$ is a small enough constant.
Again, we would like to stress that the NP-hardness bounds of~\cite{SVPhardness-RegevHavivSTOC07} for Shortest Vector do not apply for such lattices 
where a solution is guaranteed to exist. 

Finally, we can also prove that an oracle with exponentially small error for number balancing would provide approximations
for Minkowski's Theorem and Shortest Vector: 
\begin{theorem}\label{thm:npptomink} Suppose that there is a $\delta$-approximation for number balancing with 
$\delta \leq 2^{\sqrt{n}} / 2^{n}$. 
Then there is an $O(n^5)$-approximation for Minkowski's problem. Here it suffices to have a separation oracle for the convex body $K \subseteq \setR^n$.
\end{theorem}

\section{Reducing Number Balancing to Minkowski's Theorem}

In this section we will show how to solve NBP with an oracle for  Minkowski's Theorem.
The idea is to consider a hypercube intersected with the constraint $|\langle\bm{a},\bm{x}\rangle|\le \delta$, and to show that this set has large enough volume. If we have an exact Minkowski oracle, this gives us $\bm{x}\in \{-1,0,1\}^n\setminus\{\bm{0}\}$ as desired. Here we state a more general version which uses only a $\rho$-approximate Minkowski oracle, and then show how we can use this more general version to solve NBP with a weaker bound.

\begin{theorem}\label{thm:NPPoracle}
Suppose we have a $\rho$-approximate Minkowski oracle, and let $k > 0$ be any positive integer. Then, for any $\bm{a}\in [0,1]^n$, there is a polynomial time algorithm to find $\bm{x} \in \setZ^n \setminus \{\bm{0}\}$ with $\|\bm{x}\|_\infty \leq k$ and
so that $|\left<\bm{a},\bm{x}\right>| \leq n\left(\frac{\rho}{k+1}\right)^{n-1}$. 
\end{theorem}

\begin{proof}
For ease of notation, let 
\[ \delta := n\left(\frac{\rho}{k+1}\right)^{n-1}.\]
Now consider the body
\[
  K := \left\{ \bm{x} \in \left(-\frac{k+1}{\rho},\frac{k+1}{\rho}\right)^n : |\left<\bm{a},\bm{x}\right>| \leq \delta \right\}.
\]
Obviously this is a symmetric convex body. For $\alpha \in [-\delta,\delta]$, consider the $(n-1)$-dimensional slice
\[
  K(\alpha) := \left\{ \bm{x} \in \left(-\frac{k+1}{\rho},\frac{k+1}{\rho}\right)^n: \left<\bm{a},\bm{x}\right> = \alpha \right\}
\]
of it. Let us consider the $(n-1)$-dimensional volume $\mathrm{vol}_{n-1}(K(\alpha))$ of that slice.
Then we can write the volume of $K$ as
\[
  \mathrm{vol}_n(K) = \int_{-\delta}^{\delta} \mathrm{vol}_{n-1}(K(\alpha)) \; d\alpha.
\]
Moreover
\[
  \left(\frac{2(k+1)}{\rho}\right)^n = \mathrm{vol}_n\left(-\frac{k+1}{\rho},\frac{k+1}{\rho}\right)^n  = \int_{-\frac{n(k+1)}{\rho}}^{\frac{n(k+1)}{\rho}} \mathrm{vol}_{n-1}(K(\alpha)) \; d\alpha.
\]
since the slices are empty if $|\alpha| > \frac{n(k+1)}{\rho}$.
By symmetry $\mathrm{vol}_{n-1}(K(\alpha)) = \mathrm{vol}_{n-1}(K(-\alpha))$. Moreover, by convexity of $K$, 
for $\alpha \geq 0$, the quantity $\mathrm{vol}_{n-1}(K(\alpha))$ is monotonically non-increasing.
Thus
\begin{align*}
 \mathrm{vol}_n(K) &= \int_{-\delta}^{\delta} \mathrm{vol}_{n-1}(K(\alpha)) \, d\alpha \geq \frac{\delta}{\left(\frac{n(k+1)}{\rho}\right)}
  \cdot \int_{-\frac{n(k+1)}{\rho}}^{\frac{n(k+1)}{\rho}} \mathrm{vol}_{n-1}(K(\alpha)) \; d\alpha = \frac{\delta}{\left(\frac{n(k+1)}{\rho}\right)} \cdot \left(\frac{2(k+1)}{\rho}\right)^n.
\end{align*}

Now, using the fact that $\delta = n\left(\frac{\rho}{(k+1)}\right)^{n-1}$, we get $\mathrm{vol}_n(K) \geq 2^n$. Hence, using our $\rho$-approximate Minkowski oracle, we can find a vector
\[\bm{x} \in (\rho K \cap \setZ^n) \setminus \{ \bm{0}\} = \left((-(k+1),(k+1))^n \cap \setZ^n\right) \setminus \{ \bm{0}\}.\] 
In particular, this gives $\bm{x}\in \setZ^n$ with $|\langle \bm{a},\bm{x}\rangle|\le n\left(\frac{\rho}{k+1}\right)^{n-1}$ and $\|\bm{x}\|_\infty \leq k$. 
\end{proof}
The bound in Theorem~\ref{thm:NPPoracle} can be strengthened by a $\sqrt{n}$ factor by using concentration of measure
arguments. We omit this here.

Now suppose, for instance, that we have a $(2-\epsilon)$-approximate oracle for some $\epsilon \in (0,1]$. Then we can pick $k = 1$ and get
\[|\langle \bm{a},\bm{x}\rangle|\le n\left(\frac{2-\epsilon}{2}\right)^{n-1} \leq n \cdot \exp\left(-\frac{\epsilon(n-1)}{2}\right) = 2^{-\Theta(\epsilon n)},\]
with $\|\bm{x}\|_\infty \leq 1$ and $\bm x \in \mathbb Z^n \setminus \lbrace{\bm 0\rbrace}$.
However, this line of arguments breaks down if we only have access to a $\rho$-approximation for $\rho \geq 2$
as these would in general not produce feasible solutions for number balancing. 

It turns out that we can design a \emph{recursive self-reduction}. 
The main technical argument is to transform an algorithm that finds $\bm{x}\in \setZ^n\setminus\{\bm{0}\}$ with $\|\bm{x}\|_\infty\le k$, into an algorithm that finds vectors $\bm{x}\in \setZ^n\setminus\{\bm{0}\}$ with $\|\bm{x}\|_\infty\le \frac{k}{2}$, with a bounded decay in the error $|\langle\bm{a},\bm{x}\rangle|$. Applying this recursively gives the following lemma.
\begin{lemma}\label{lem:fullselfreduction}
Suppose that there is a polynomial-time algorithm that for any $\bm{a}' \in [-1,1]^n$ finds a 
vector $\bm{x}' \in \{ -k,\ldots,k\}^n \setminus \{ \bm{0} \}$ with $|\left<\bm{a}',\bm{x}'\right>| \leq 2^{-n}$. 
If $k\leq \frac {\log n} {6\log \log n}$,
then there is also a polynomial time algorithm 
that for any $\bm{a} \in [-1,1]^n$ finds 
a vector $\bm{x} \in \setZ^n$ with $|\left<\bm{a},\bm{x}\right>| \leq 2^{-n^{\frac 1 {3k}}}$
and $\bm{x}\in \{ -1,0,1\}^n \setminus \{ \bm{0} \}$.
\end{lemma}

The way we do the self-reduction is the following.
We partition our set of $n$ numbers into subsets of size $\sqrt{n}$.
First, for each subset $\ell$, we find a number $b_\ell\ne 0$ for which we can (approximately) express 
$b_\ell,2b_\ell,\ldots,kb_\ell$
as linear combinations of elements of that subset using only coefficients in $\{-\lfloor \frac{k}{2}\rfloor,\ldots\lfloor\frac{k}{2}\rfloor\}$.
We then run our assumed algorithm on $b_1,\ldots,b_{\sqrt{n}}$ to obtain $\bm{y}\in \{-k,\ldots,k\}^n$ with $\langle \bm{b},\bm{y}\rangle=\sum_{\ell=1}^{\sqrt{n}}y_\ell b_\ell$ being small.
Since each of the summands can be expressed more efficiently in terms of our original set of numbers, we obtain a good solution $\bm{x}$ with coefficients in $\{-\lfloor \frac{k}{2}\rfloor,\ldots\lfloor\frac{k}{2}\rfloor\}$. 

%
%
%
The following two lemmas go through this argument more precisely. 
The first gives a condition under which we can find a more efficient coefficient representation for a number $\beta$.
In the second, we will show that an application of the assumed algorithm will actually allow us to satisfy this condition, and so we are able to construct the stronger oracle.
Note that the interesting parameter 
choice is $r := \lceil k/2 \rceil$, so that the size of the coefficients is halved.

\begin{lemma} \label{lem:RepresentationWithSmallCoefficients}
Let $r,k \in \setN$ be parameters with $r<k$.
Let $\alpha_1,\ldots,\alpha_k \in \setR$ so that $\sum_{i=1}^k i \cdot \alpha_i = 0$
and abbreviate $\beta := \alpha_{r} + \ldots + \alpha_k$. 
Then $\{ j \cdot \beta \mid j=-k,\ldots,k\} \subseteq \{ \sum_{i=1}^k \lambda_i \alpha_i \mid \lambda_i \in \setZ\textrm{ and }|\lambda_i| \leq \max\{ r-1,k-r\}\textrm{ for }i=1,\ldots,k\}$.
\end{lemma}
\begin{proof}
By symmetry it suffices to consider $j \geq 0$.
For $j \in \{ 0,\ldots,r-1\}$, we can obviously write 
\[
j \cdot \beta = j \cdot \alpha_{r} + \ldots + j \cdot \alpha_k.
\]
Now consider $j \in \{ r,\ldots,k\}$. The trick is to use that
\[
r \cdot \beta + \sum_{i=1}^{r-1} i \cdot \alpha_i + \sum_{i=r}^k (i-r) \cdot \alpha_i = \sum_{i=1}^k i \cdot \alpha_i = 0 \quad \Rightarrow 
  \quad r \cdot \beta = \sum_{i=1}^{r-1} (-i) \cdot \alpha_i + \sum_{i=r}^{k} (r-i) \cdot \alpha_i. 
\]
Then
\[
  j \cdot \beta = (j-r) \cdot \beta + r \cdot \beta = \sum_{i=1}^{r-1} (-i) \cdot \alpha_i + \sum_{i=r}^{k} \underbrace{[(j-r)+(r-i)]}_{=j-i} \cdot \alpha_i 
\]
If we inspect the size of the used coefficients, then for $i \in \{ 1,\ldots,r-1\}$ we have $|-i| \leq r-1$ and for $i \in \{ r,\ldots,k\}$ we have $|j-i| \leq k-r$. 
\end{proof}
Note that if we have the weaker assumption of $|\sum_{i=1}^k i \cdot \alpha_i| \leq \delta$, then 
for any $j \in \{ 0,\ldots,k\}$ one can find coefficients $\lambda_i$ with $|j \cdot \beta - \sum_{i=1}^k \lambda_i \alpha_i| \leq \delta$.

\begin{lemma}\label{lem:halvecoefficient}
Let $k,r \in \setN$ be parameters with $0<r<k$.
Let $f:\mathbb{N}\rightarrow {R}$ be a non-negative function such that $f(n)\geq 4\log n$.
Suppose that there is a polynomial-time algorithm that for any $\bm{a}' \in [-1,1]^n$ finds a 
vector $\bm{x}' \in \{ -k,\ldots,k\}^n \setminus \{ \bm{0} \}$ with $|\left<\bm{a}',\bm{x}'\right>| \leq 2^{-f(n)}$. Then there is also a polynomial time algorithm 
that for any $\bm{a} \in [-1,1]^n$ finds 
a vector $\bm{x} \in \setZ^n\setminus \{ \bm{0} \}$ with $|\left<\bm{a},\bm{x}\right>| \leq 2^{-f(\sqrt{n})/2}$
and $\|\bm{x}\|_{\infty} \leq \max\{ r-1,k-r\}$.
\end{lemma}
\begin{proof}
Let $\bm{a} \in [-1,1]^n$ be the given vector of numbers.
Split $[n]$ into blocks $I_1,\ldots,I_{\sqrt{n}}$ each of size $|I_{\ell}| = \sqrt{n}$.
For each block $I_{\ell}$ we use the oracle to find a vector $\bm{x}_{\ell} \in \{ -k,\ldots,k\}^n\setminus\{\bm{0}\}$ with
$\textrm{supp}(\bm{x}_{\ell}) \subseteq I_{\ell}$ so that $|\left<\bm{a},\bm{x}_{\ell}\right>| \leq 2^{-f(\sqrt{n})}$. If for any $\ell$ one has $\|\bm{x}_{\ell}\|_{\infty} \leq r-1$, 
then we simply return $\bm{x} := \bm{x}_{\ell}$ and are done.
Otherwise, we write the vector as $\bm{x}_{\ell} = \sum_{i=1}^k i \cdot \bm{x}_{\ell,i}$ with
vectors $\bm{x}_{\ell,1},\ldots,\bm{x}_{\ell,k} \in \{ -1,0,1\}^n$. 
Note that these vectors will have disjoint support and $\textrm{supp}(\bm{x}_{\ell,1}),\ldots,\textrm{supp}(\bm{x}_{\ell,k}) \subseteq I_{\ell}$.
Moreover we know that for every $\ell$ there is at least one index $i \in \{ r,\ldots,k\}$
with $\bm{x}_{\ell,i} \neq \bm{0}$.

Now define a vector $\bm{b} \in \setR^{\sqrt{n}}$ with $b_{\ell} := \sum_{i=r}^k \left<\bm{a},\bm{x}_{\ell,i}\right>$. 
Note that if for any $\ell$ we have $|b_\ell|\le 2^{-f(\sqrt{n})}$, then we can set $\bm{x}=\sum_{i=r}^k\bm{x}_{\ell,i}$ and we are done. Therefore we may assume that $|b_\ell|>2^{-f(\sqrt{n})}$ for all $\ell$.
Also note that $\|\bm{b}\|_{\infty} \leq \sqrt{n}$.
We run the oracle again to find a vector $\bm{y} \in \{ -k,\ldots,k\}^{\sqrt{n}}\setminus\{\bm{0}\}$ so that $|\left<\bm{b},\bm{y}\right>| \leq \sqrt{n} \cdot 2^{-f(\sqrt{n})}$. For each block $\ell \in [\sqrt{n}]$ we can use Lemma~\ref{lem:RepresentationWithSmallCoefficients} to find integer coefficients $\lambda_{\ell,i}$ with  $|\lambda_{\ell,i}| \leq \max\{ r-1,k-r\}$ so that
\[
  \Big| y_{\ell} \cdot b_{\ell} - \sum_{i=1}^k \lambda_{\ell,i} \cdot \left<\bm{a},\bm{x}_{\ell,i}\right> \Big| \leq 2^{-f(\sqrt{n})}.
\]
We define
\[
  \bm{x} := \sum_{\ell=1}^{\sqrt{n}} \sum_{i=1}^k \lambda_{\ell,i} \bm{x}_{\ell,i}.
\]
Then $\|\bm{x}\|_{\infty} \leq \max\{ r-1,k-r\}$ since the $\bm{x}_{\ell,i}$'s have disjoint support and $\|\bm{x}_{\ell,i}\|_{\infty} \leq 1$ for all $\ell,i$. 
Moreover, since there is some $y_\ell\ne 0$ and $|b_\ell|>2^{-f(\sqrt{n})}$, we have $\bm{x} \neq \bm{0}$. 

Finally we inspect that
\[
  |\left<\bm{a},\bm{x}\right>| \leq |\left<\bm{y},\bm{b}\right>| + \sum_{\ell=1}^{\sqrt{n}} \Big|y_{\ell}b_{\ell} - \sum_{i=1}^k \lambda_{\ell,i} \left<\bm{a},\bm{x}_{\ell,i}\right>\Big| \le 2\sqrt{n}\cdot 2^{-f(\sqrt{n})} \leq 2^{-f(\sqrt{n})/2}.
\]

The last line comes from the fact that when $f(n)\geq 4\log n$, we have
$2\sqrt{n}\leq 2^{\frac 2 3\log n}= 2^{\frac 4 3\log \sqrt{n}}\leq 2^{\frac 1 3f(\sqrt{n})}$. 
\end{proof}

Now we can apply Lemma \ref{lem:halvecoefficient} recursively to prove Lemma \ref{lem:fullselfreduction}.
\begin{proof}[Proof of Lemma \ref{lem:fullselfreduction}]
For $k\in \setN$,
take $r=\lceil k/2\rceil$,
then Lemma \ref{lem:halvecoefficient} says any $\bm{a} \in [-1,1]^n$,
if we have a $\bm{x} \in \{ -k,\ldots,k\}^n$ such that $|\left<\bm{a},\bm{x}\right>| \leq 2^{-f(n)}$,
we can obtain $\bm{x'} \in \{ -\lfloor k/2\rfloor,\ldots,\lfloor k/2\rfloor\}^n$ such that $|\left<\bm{a},\bm{x}'\right>| \leq 2^{-f(\sqrt{n})/2}$.
Repeat this procedure for $t:=\lceil \log k \rceil$ times,
if $f(n^{2^{-t}})\geq 4\log n$,
we will get $\bm{x}^{''} \in \{ -1,0,1\}^n$ such that $|\left<\bm{a},\bm{x}''\right>| \leq 2^{-f(n^{2^{-t}})/2^t}$.

So we can take $f(n)=n$ as appearing in Lemma \ref{lem:halvecoefficient}.
Note that when $k\leq \frac{\log n}{6\log \log n}$,
$t=\lceil \log k \rceil\leq \log (\frac 1 2\frac{\log n}{\log \log n})$,
so \[
f(n^{2^{-t}})=n^{2^{-t}}\geq n^{2^{-\log (\frac 1 2\frac{\log n}{\log \log n})}}=\log^2 n\geq 4\log n.
\]
Then the condition in Lemma \ref{lem:halvecoefficient} is satisfied, and the bound we get is
\[
2^{-f(n^{2^{-t}})/2^t}=2^{-(n^{2^{-t}})/2^t}\leq 2^{-\frac{n^{\frac 1 {2k}}}{2k}}\leq 2^{-n^{\frac 1{3k}}}.
\]
Here we use the fact that $t=\lceil \log k \rceil\leq \log 2k$,
and when $k\leq \frac 1 {6}\frac{\log n}{\log \log n}$, we have $\frac{n^{\frac 1 {2k}}}{2k}\geq n^{\frac 1{3k}}$. 
\end{proof}

Using Lemma \ref{lem:fullselfreduction}, we are now able to prove Theorem \ref{thm:minktonpp}:

\begin{proof}[Proof of Theorem \ref{thm:minktonpp}]
 If $\rho\geq \frac {\log n} {48\log\log n}$,
 then $2^{-n^{\Theta(1/\rho)}}=2^{-\log^{O(1)}}n$.
 By choosing a proper constant on the exponent,
 this can be achieved with the Karmarkar-Karp algorithm.
 So we only need to work with $\rho< \frac {\log n} {48\log\log n}$.

 If we take $k=3\rho$, 
 then Theorem \ref{thm:NPPoracle} implies $|\left<\bm{a},\bm{x}\right>| \leq n\left(\frac{\rho}{k+1}\right)^{n-1}\leq 2^{-n}$.
 Moreover,
 $k=3\rho\leq\frac{\log n}{16\log\log n}$,
and hence the condition of Lemma \ref{lem:fullselfreduction} is satisfied.
 Then the bound given by Lemma \ref{lem:fullselfreduction} is $2^{-n^{\frac 1 {3k}}}\leq 2^{-n^{\Theta(1/\rho)}}$. 
 \end{proof}

Instead of using an oracle for Minkowski's Theorem one can directly use an oracle for the Shortest Vector Problem in the $\|\cdot\|_{\infty}$ norm. We need the following theorem:
\begin{theorem}\label{thm:SVPoracle}
Suppose that there is a polynomial time algorithm that given a lattice 
$\Lambda \subseteq \setR^n$ with $\det(\Lambda) \leq 1$ finds a non-zero vector $\bm{x} \in \Lambda$ of length $\|\bm{x}\|_{\infty} \leq \rho$.
Then, for any $\bm{a}\in [0,1]^n$, there is a polynomial time algorithm to find $\bm{x} \in \setZ^n \setminus \{\bm{0}\}$ with $\|\bm{x}\|_\infty \leq k$ and
so that $|\left<\bm{a},\bm{x}\right>| \leq 2nk\rho(\frac{\rho}{k})^{n}$. 
\end{theorem}
\begin{proof}
When $\rho(\frac{\rho}{k})^{n}\geq \frac 1 2$, we have 
$2nk\rho(\frac{\rho}{k})^{n}\geq 1$, in which case $\bm{x}=(1,0,\cdots,0)$ trivially does the job. We can hence assume that $\rho(\frac{\rho}{k})^{n}< \frac 1 2$.

Let $\bm I_n$ be the $n\times n$ identity matrix and consider the lattice $\Lambda$ generated by the $(n+1) \times (n+1)$-dimensional matrix
$$B:=
\begin{pmatrix} \frac {\rho}{k}\bm{I_n} & \bm{0} \\ \frac 1 {2nk}(\frac{k}{\rho})^{n}\bm{a}^T & (\frac k {\rho})^{n}\end{pmatrix}.$$

Note that $\det(B)=1$.
Let  $\bm{x} \in \Lambda$ be the vector returned by the algorithm.
Then $\bm{y}:=B^{-1}\bm{x}\in \setZ^{n+1}$.
Since $\|\bm{x}\|_{\infty} \leq \rho$ and $\bm{x}=B\bm{y}$,
we have $|y_i|\leq k$ for $i=1,\ldots,n$.
Moreover,
\[
\begin{split}
 \rho\geq |x_{n+1}|&=\left|y_{n+1}\left(\frac k {\rho}\right)^{n}+\sum_{i=1}^n\frac 1 {2nk}\left(\frac{k}{\rho}\right)^{n}y_ia_i\right|  \geq |y_{n+1}|\left(\frac k {\rho}\right)^{n}-\left|\sum_{i=1}^n\frac 1 {2nk}\left(\frac{k}{\rho}\right)^{n}y_ia_i\right|\\
 &\overset{|y_i| \leq k}{\geq} |y_{n+1}|\left(\frac k {\rho}\right)^{n}-\frac 1 {2}\left(\frac{k}{\rho}\right)^{n},
\end{split}
\]
which implies $|y_{n+1}|\leq \frac 12+\rho(\frac{\rho}{k})^{n}$. Since $\rho(\frac{\rho}{k})^{n}<\frac 1 2$ and $y_{n+1} \in \mathbb Z$, we must have $y_{n+1}=0$.

Therefore,
$\left|\sum_{i=1}^n\frac 1 {2nk}\left(\frac{k}{\rho}\right)^{n}y_ia_i\right| = |x_{n+1}| \leq \rho$, and hence $\left|\sum_{i=1}^ny_ia_i\right| \leq 2nk\rho\left(\frac{\rho}{k}\right)^n$, so the vector $(y_1,\ldots,y_n)$ does the job.

\end{proof}

Now using Lemma \ref{lem:fullselfreduction} and Theorem \ref{thm:SVPoracle} we are able to prove Theorem \ref{thm:ReducingNBPToSVP}:

\begin{proof}[Proof of Theorem \ref{thm:ReducingNBPToSVP}] The proof mirrors that of Theorem \ref{thm:minktonpp}.  We start by choosing $k = 3\rho$ in Theorem \ref{thm:SVPoracle}, which gives us $|\left<\bm{a},\bm{x}\right>| \leq 2^{-n}$ and $\|\bm{x}\|_\infty \leq 3\rho$. As we argued in the proof of Theorem \ref{thm:minktonpp}, we may assume that $\rho< \frac {\log n} {48\log\log n}$. This gives us $k \leq \frac{\log n}{16 \log \log n}$ and we can, as before, apply Lemma \ref{lem:fullselfreduction}. 
\end{proof}

\section{Reducing Minkowski's Theorem to Number Balancing}

In this section we show how an oracle for number balancing can be used to design an algorithm to
approximate Minkowski's problem. More precisely, given a large enough symmetric convex body $K$, 
we will be able to find a vector $\bm{x}\in \setZ^n\cap \rho K\setminus\{\bm{0}\}$ where $\rho$ is a polynomial in $n$.
The first helpful insight is that each such convex body can be approximated within a factor 
of $\sqrt{n}$ using an \emph{ellipsoid}~\cite{lovasz1990geometric,grotschel2012geometric}. Recall that an ellipsoid is a set of the form
\begin{equation} \label{eq:Ellipsoid}
  \pazocal{E} = \Big\{ \bm{x} \in \setR^n \mid \sum_{i=1}^n \frac{1}{\lambda_i^2}\cdot \left<\bm{x},\bm{a}_i\right>^2 \leq 1 \Big\}
\end{equation}
with an \emph{orthonormal basis} $\bm{a}_1,\ldots,\bm{a}_n \in \setR^n$ defining the \emph{axes} and positive coefficients $\lambda_1,\ldots,\lambda_n > 0$ that describe the \emph{lengths of the axes}\footnote{Strictly speaking, the length of axis $i$ is $2\lambda_i$, but we will continue calling $\lambda_i$ the ``axis length''.}.
Overall, our reduction will operate in two steps: 
\begin{enumerate}
\item[(i)] By combining \emph{John's Theorem} with \emph{lattice basis reduction}, we can show that it suffices to find integer points in an ellipsoid that is \emph{well-rounded}, meaning that the lengths of the axes are bounded.
\item[(ii)] We show that a number balancing oracle allows a self-reduction to a generalized form where inner products
with $n$ vectors have to be minimized and additionally the solution space is $\setZ^n$ instead of $\{ -1,0,1\}^n$.  
\end{enumerate}
We begin by proving $(ii)$ and postpone $(i)$ until the end of this section.

\subsection{A self-reduction to a generalized form of number balancing}



The main technical result of this section is the following reduction:
\begin{theorem} \label{thm:GeneralNBPtoNBP} 
Suppose there is a $\delta$-approximation for number balancing with error 
parameter $\delta=\frac{\rho(n)}{2^n}$ and $\rho(n) \leq 2^{n/2}$.
Then there is a polynomial time algorithm that on input $\bm{a}_1,\ldots,\bm{a}_n\in [0,1]^n$ and $0<\lambda_1\le \ldots \le \lambda_n\le 2^{n}$ with $\prod_{i=1}^n\lambda_i\ge 1$, finds a vector $\bm{x}\in \setZ^n\setminus\{\bm{0}\}$ with
$$|\langle \bm{x},\bm{a}_i\rangle | \le O(n^4) \cdot \lambda_i\cdot \rho(4n^2)^{1/n} \quad \forall i=1,\ldots,n.$$
\end{theorem}
In particular if $\rho(n) \leq 2^{\sqrt{n}}$, then the right hand side in Theorem~\ref{thm:GeneralNBPtoNBP}
simplifies to just $O(n^4) \cdot \lambda_i$.
We will show this by introducing two extensions of the number balancing oracle. The first extension 
gives a weaker bound in terms of the error parameter, but allows for multiple vectors in $[0,1]^n$.
In the second extension, we extend the range of coefficients from $\{-1,0,1\}$ to  $\{-Q,\ldots,Q\}$ 
which leads to a much stronger error bound. 

\begin{lemma}\label{lem:kvectors} 
Suppose there is a $\delta$-approximation for number balancing. 
Then there is a polynomial time algorithm that given an input $\bm{a}_1,\ldots,\bm{a}_k \in [0,1]^n$
and parameters  $\delta_1,\ldots,\delta_k\le \frac{1}{2}$ with $\prod_{i=1}^k\delta_i\ge \delta$ 
finds a vector $\bm{x}\in \{-1,0,1\}^n\setminus\{\bm{0}\}$ with $|\langle \bm{a}_i,\bm{x}\rangle|\le 2n^2\delta_i$ for all $i=1,\ldots,k$.
\end{lemma}

\begin{proof}
The idea is that we will discretize all of the vectors and then run our oracle on their sum.
The vector that we obtain will then have small inner product with all of the $\bm{a}_i$.
To define the discretization $\tilde{\bm{a}}_i$, round elements of $\bm{a}_i$ down to the nearest multiple of $2n\delta_i$, and then multiply by $\prod_{j<i}\delta_j$.
Defining $\tilde{\bm{a}}_i$ this way, notice that for all $i$ we have $|\langle \tilde{\bm{a}}_i,\bm{x}\rangle|\le n\prod_{j<i}\delta_j$.

Now let $\bm{c}=\tilde{\bm{a}}_1+\ldots+\tilde{\bm{a}}_k$. 
By our oracle, we can find $\bm{x}\in \{-1,0,1\}^n\setminus\{\bm{0}\}$ with $|\langle \bm{c},\bm{x}\rangle|\le \delta$.
We have
\begin{eqnarray*}
|\langle \tilde{\bm{a}}_1,\bm{x}\rangle|
&\le&
|\langle \tilde{\bm{a}}_2,\bm{x}\rangle|+\ldots+|\langle\tilde{\bm{a}}_k,\bm{x}\rangle|+|\langle \bm{c},\bm{x}\rangle|
\le 
n\delta_1+n\prod_{j\le 2}\delta_j+\ldots+n\prod_{j\le k}\delta_j\\
&\le & 
n\delta_1\cdot \left(1+\frac{1}{2}+\ldots+\frac{1}{2^k}\right) <  2n\delta_1.
\end{eqnarray*}
Therefore $|\langle \tilde{\bm{a}}_1,\bm{x}\rangle|=0$. 
Similarly, for $1<i\le k$, if $|\langle\tilde{\bm{a}}_1,\bm{x}\rangle|,\ldots,|\langle\tilde{\bm{a}}_{i-1},\bm{x}\rangle|=0$, then we have
$$|\langle \tilde{\bm{a}}_i,\bm{x}\rangle|\le 
|\langle \tilde{\bm{a}}_{i+1},\bm{x}\rangle|+\ldots+|\langle\tilde{\bm{a}}_k,\bm{x}\rangle| + |\langle \bm{c},\bm{x}\rangle|
< 2n\prod_{j\le i}\delta_j,$$
and hence $|\langle \tilde{\bm{a}}_i,\bm{x}\rangle|=0$ for all $i$.
Notice that by definition of $\tilde{\bm{a}}_i$ we have $\|\prod_{j<i}\delta_j \bm{a}_i-\tilde{\bm{a}}_i\|_\infty\le 2n\prod_{j\le i}\delta_j$.
Therefore $|\langle\prod_{j<i}\delta_j \bm{a}_i,\bm{x}\rangle|\le 2n^2\prod_{j\le i}\delta_j$, and so we can conclude that
$|\langle \bm{a}_i,\bm{x}\rangle|\le 2n^2\delta_i$. 
\end{proof}

Now we come to a second reduction that takes the oracle constructed in Lemma~\ref{lem:kvectors}
as a starting point: 
\begin{lemma}\label{lem:addingQ} 
Assume there exists a $f(n)$-approximation for number balancing.
Let $\bm{a}_1,\ldots,\bm{a}_k\in [0,1]^n$ be given with parameters $\delta_1,\ldots,\delta_k\le \frac{1}{2}$ and a number $Q$
that is a power of $2$ and satisfies $\prod_{i=1}^k\delta_i\ge f(n\log{Q})$. Then in polynomial time we can find
a vector $\bm{x}\in \{-Q,\ldots,Q\}^n\setminus\{\bm{0}\}$ with 
$|\langle \bm{a}_i,\bm{x}\rangle|\le \delta_i  Q\cdot 2(n\log{Q})^2$ for all $i=1,\ldots,k$.
\end{lemma}

\begin{proof} 
For each $i$, we define $\bm{b}_i\in [0,1]^{n\log{Q}}$ by $\bm{b}_i(j,\ell)=\bm{a}_i(j)2^{-\ell}$ for $j=1,\ldots,n$ and $\ell=1,\ldots,\log{Q}$.
Since $\prod_{i=1}^k\delta_i\ge f(n\log{Q})$,
we can apply Lemma \ref{lem:kvectors} to find $\bm{y}\in \{-1,0,1\}^{n\log{Q}}\setminus\{\bm{0}\}$ with
$|\langle \bm{b}_i,\bm{y}\rangle |\le \delta_i\cdot 2(n\log{Q})^2$.

Now define $\bm{x}\in \{-Q,\ldots,Q\}^n\setminus\{\bm{0}\}$ by $x_j:=Q\sum_{\ell=1}^{\log{Q}}2^{-\ell}y_{j\ell}$.
Then for $i=1,\ldots,k$ we have
\[
 \delta_i \cdot (2n \log Q)^2 \geq \left|\left<\bm{b}_i,\bm{y}\right>\right| = \Big|\sum_{j=1}^n \underbrace{\sum_{\ell=1}^{\log Q} y_{j\ell} 2^{-\ell}}_{=x_j/Q} a_i(j)\Big| = \frac{1}{Q} \cdot \left|\left<\bm{a}_i,\bm{x}\right>\right|
\]
and rearranging gives the claim. 
\end{proof}

Finally we come to the proof of Theorem~\ref{thm:GeneralNBPtoNBP}.  

\begin{proof}[Theorem~\ref{thm:GeneralNBPtoNBP}]
Suppose that the oracle has parameter $f(n) = \rho(n) / 2^n$.
Suppose that $\bm{a}_1,\ldots,\bm{a}_n\in [0,1]^n$ and $\lambda_1,\ldots,\lambda_n>0$ with $\prod_{i=1}^n\lambda_i\ge 1$.
We choose $Q := 2^{4n}$, which is a power of 2. Define $\delta_i=\lambda_i\cdot f(n\log{Q})^{1/n}$. Note that
$\delta_i \leq 2^{3n/2} \cdot f(4n^2)^{1/n} \le \frac{1}{2}$ since $f(n) \leq 2^{-n/2}$.
Then $\prod_{i=1}^n\delta_i\ge f(n\log{Q})$, and so 
by Lemma \ref{lem:addingQ} we can find $\bm{y}\in \{-Q,\ldots,Q\}^n\setminus\{\bm{0}\}$ with 
\begin{eqnarray*}
|\langle \bm{a}_i,\bm{y}\rangle|\le Q\delta_i\cdot 2(n\log{Q})^2&\le& Q\lambda_i\cdot f(n\log{Q})^{1/n}  \cdot 2(n\log{Q})^2 \\
&=& \lambda_i \cdot \rho(4n^2)^{1/n} \cdot 2\cdot (4n^2)^2. 
\end{eqnarray*}
\end{proof}


\subsection{A reduction to well-rounded ellipsoids}
Using John's Theorem \cite{JohnsTheorem1948}, the convex body $K$ in Theorem~\ref{thm:npptomink} can be approximated by 
an ellipsoid $\pazocal{E}$ as defined in Eq.~\eqref{eq:Ellipsoid}. The natural approach will then be to apply
Theorem~\ref{thm:GeneralNBPtoNBP} to the axes of the ellipsoid. However, it will be crucial that
the lengths of the axes of the ellipsoid are bounded by $2^{O(n)}$. We will now argue how to
make an arbitrary ellipsoid well rounded. 

Let us denote $\lambda_{\max}(\pazocal{E}) := \max\{ \lambda_i : i=1,\ldots,n \}$ as the maximum length of an axis. 
Recall that a matrix $U \in \setR^{n \times n}$ is \emph{unimodular} if $U \in \setZ^{n \times n}$
and $|\textrm{det}(U)| = 1$. In particular, the linear map $T : \setR^n \to \setR^n$ with 
$T(\bm{x}) = U\bm{x}$ is a bijection on the integer lattice, meaning that $T(\setZ^n) = \setZ^n$.
It turns out that one can use the \emph{lattice basis reduction} method 
by Lenstra, Lenstra and Lov{\'a}sz~\cite{lenstra1982factoring} to find a unimodular linear transformation that ``regularizes''
any given ellipsoid. 
Note that it suffices to work with the regularized ellipsoid $T(\pazocal{E})$
since $\textrm{vol}_n(T(\pazocal{E})) 
 = \textrm{vol}_n(\pazocal{E})$ and if we find a point $\bm{x} \in (\rho T(\pazocal{E})) \cap \setZ^n$, then 
by linearity $T^{-1}(\bm{x}) \in \rho \pazocal{E}$ and $T^{-1}(\bm{x}) \in \setZ^{n}$.

%

Given $\bm{b}_1,\ldots,\bm{b}_n\in \setR^n$ we define the \emph{Gram-Schmidt orthogonalization} iteratively as $
 \hat{\bm{b}}_j=\bm{b}_j-\sum_{i<j}\mu_{ij}\hat{\bm{b}}_i,$ where $\mu_{ij}=\frac{\langle \bm{b}_j,\hat{\bm{b}}_i\rangle}{\|\hat{\bm{b}}_i\|_2^2}.$
Notice that we can then write $\bm{b}_j=\hat{\bm{b}}_j+\sum_{i<j}\mu_{ij}\hat{\bm{b}}_i.$
In particular, suppose $B$ is the matrix with columns $\bm{b}_1,\ldots,\bm{b}_n$ and $\hat{B}$ is the matrix with columns $\hat{\bm{b}}_1,\ldots,\hat{\bm{b}}_n$. Then $B=\hat{B}V$ for an upper triangular matrix $V$ with ones along the diagonal and $V_{ij}=\mu_{ij}$ for $i<j$.

\begin{definition} Let $B\in \setR^{n\times n}$ be a lattice basis and let $\mu_{ij}$ be the coefficients from Gram-Schmidt orthogonalization. The basis is called \emph{LLL reduced} if 
\begin{itemize}
\item (Coefficient-reduced): $|\mu_{ij}|\le \frac{1}{2}$ for all $1\le i<j\le n$.
\item (Lov{\'a}sz condition): $\|\hat{\bm{b}}_i\|_2^2\le 2\|\hat{\bm{b}}_{i+1}\|_2^2$ for $i=1,...,n-1$.
\end{itemize}
\end{definition}
LLL reduction has been widely used in diverse fields such as integer programming and cryptography \cite{nguyen2010lll}.
One property of the LLL reduced basis 
is that the eigenvalues of the corresponding matrix $B$ are bounded away from $0$: 

\begin{lemma}\label{lem:LLLlowerbound} Let $B$ denote the matrix with columns $\bm{b}_1,\ldots,\bm{b}_n$. If $\bm{b}_1,\ldots,\bm{b}_n$ is an $LLL$-reduced basis with $\|\bm{b}_i\|_2\ge 1$ for all $i$, then $\|B\bm{x}\|_2\ge 2^{-3n/2}\cdot \|\bm{x}\|_2$ for all $\bm{x}\in \setR^n$.
\end{lemma}

\begin{proof}
Let $\hat{B}$ denote the Gram-Schmidt orthogonalization of $B$, with columns $\hat{\bm{b}}_1,\ldots,\hat{\bm{b}}_n$.
%
Now, for any $k$, we can use the properties of LLL reduction to gain the following bound~(proposition (1.7) in \cite{lenstra1982factoring}).  
$$1\le \|\bm{b}_k\|_2^2=\|\hat{\bm{b}}_k\|_2+\sum_{i<k}\mu_{ik}^2\|\hat{\bm{b}}_i\|_2^2\le \Big( 1+\frac{1}{4}\sum_{i<k}2^{k-i}\Big) \cdot \|\hat{\bm{b}}_k\|_2^2\le 2^k\cdot \|\hat{\bm{b}}_k\|_2^2.$$
In particular, $\|\hat{\bm{b}}_k\|^2\ge 2^{-n}$ for all $k=1,\ldots,n$.
Now let $V$ be the matrix so that $B=\hat{B}V$, and let $\bm{x}\in \setR^n$ with $\|\bm{x}\|_2=1$. 
Let $k$ denote the largest index with $|x_k|\ge 2^{-k}$.
Then 
$$|(V\bm{x})_k|= \Big|x_k+\sum_{j>k}\mu_{kj}x_j\Big| \ge |x_k|-\frac{1}{2}\sum_{j>k}|x_j|
\ge 2^{-k}-\frac{1}{2}\sum_{j>k}2^{-j}\ge 2^{-n}.$$
Now, by the orthogonality of $\hat{\bm{b}}_1,\ldots,\hat{\bm{b}}_n$, we have 
$$\|B\bm{x}\|_2^2=\|\hat{B}V\bm{x}\|_2^2=\sum_{i=1}^n(V\bm{x})_i^2\|\hat{\bm{b}}_i\|_2^2 \ge|(V\bm{x})_k|^2 \cdot 2^{-n}\ge 2^{-3n}.$$
Taking square roots gives the claim. 
%
\end{proof}

\begin{lemma} \label{lem:FindingUnimodularTransformation}
Let $\pazocal{E}=\{\bm{x} \in \setR^n : \|A\bm{x}\|_2^2\le 1\}$ be an ellipsoid. 
Then in polynomial time, we can find
\begin{enumerate}
\item[(1)] either a vector $\bm{x}\in \pazocal{E}\cap \setZ^n$ 
\item[(2)] or a linear transformation $T$ so that 
 $T(\bm{x})=U\bm{x}$ for a unimodular matrix $U$ and $\lambda_{\textrm{max}}(T(\pazocal{E}))\le 2^{3n/2}$.
\end{enumerate}
\end{lemma}

\begin{proof}
Use the algorithm of \cite{lenstra1982factoring} to find 
a unimodular matrix $U$ such that 
$B=AU$ is LLL reduced. Let $\bm{b}_1,\ldots,\bm{b}_n$ denote the columns of $B$.
Notice that if $\|\bm{b}_i\|_2\le 1$, then 
$A^{-1}\bm{b}_i\in \pazocal{E}\cap \setZ^n$, and so we are done.
So assume now that $\|\bm{b}_i\|_2\ge 1$ for all $i$. 

Define $T(\bm{x})=U^{-1}\bm{x}$, and notice that $T(\pazocal{E})=\{\bm{x} \in \setR^n :\|B\bm{x}\|_2^2\le 1\}$.
We then have
$$\lambda_{\textrm{max}}(f(\pazocal{E}))=\max_{\bm{x}\in f(\pazocal{E})}\|\bm{x}\|_2=\max_{\|B\bm{x}\|_2\le 1}\|\bm{x}\|_2=\max_{\bm{x}\ne 0}\frac{\|\bm{x}\|_2}{\|B\bm{x}\|_2}\le 2^{3n/2},$$
where the last inequality follows from Lemma \ref{lem:LLLlowerbound}. 
%
\end{proof}

%
Finally we can prove one of our main results, Theorem~\ref{thm:npptomink}. 
\begin{proof}[Theorem~\ref{thm:npptomink}]
Let $K \subseteq \setR^n$ be a convex body with $\mathrm{vol}_n(K) \geq 2^n$. We compute an ellipsoid\footnote{Note that there \emph{exists} an ellipsoid that approximates $K$ within a factor of $\sqrt{n}$ and if $K$ is a polytope with $m$ facets, then this ellipsoid can be found in time polynomial in $n$ and $m$. However, if one only has a separation oracle for $K$, then the best factor achievable in polynomial time is $n$.} 
$\pazocal{E} = \{ \bm{x} \in \setR^n \mid \sum_{i=1}^n \frac{1}{\lambda_i^2} \left<\bm{x},\bm{a}_i\right>^2 \leq 1\}$
so that $\frac{1}{5\sqrt{n}} \pazocal{E} \subseteq K \subseteq \frac{1}{5}\sqrt{n} \pazocal{E}$. Then 
\[
  2^n \cdot 5^n \cdot n^{-n/2}\leq \mathrm{vol}_n(K) \cdot 5^n \cdot n^{-n/2} \leq \mathrm{vol}_n(\pazocal{E}) = \underbrace{\mathrm{vol}(B(\bm{0},1))}_{\leq 5^n n^{-n/2}} \cdot \prod_{i=1}^n \lambda_i.
\]
and hence $\prod_{i=1}^n \lambda_i \geq 1$.
We apply Lemma~\ref{lem:FindingUnimodularTransformation} to either find an integer point in $\pazocal{E}$
and we are done, or we find a unimodular transformation $T$ so that the ellipsoid $\tilde{\pazocal{E}} := T(\pazocal{E})$
has all axes of length at most $2^{O(n)}$. Suppose the latter case happens. We write 
$\tilde{\pazocal{E}} = \{ \bm{x} \in \setR^n\mid \sum_{i=1}^n \frac{1}{\tilde{\lambda}_i^2} \left<\bm{x},\tilde{\bm{a}}_i\right>^2 \leq 1\}$ and observe that still $\prod_{i=1}^n \tilde{\lambda}_i \geq 1$ as the volume of the ellipsoid has not changed. We make use of the $\delta$-approximation for the number balancing
problem to apply Theorem~\ref{thm:GeneralNBPtoNBP} to the vectors $\tilde{\bm{a}}_1,\ldots,\tilde{\bm{a}}_n$ and parameters $\tilde{\lambda}_1,\ldots,\tilde{\lambda}_n$
and obtain a vector $\bm{x} \in \setZ^n \setminus \{ 0\}$ with $|\left<\tilde{\bm{a}}_i,\bm{x}\right>| \leq \tilde{\lambda}_i \cdot O(n^4)$. 
Then $\sum_{i=1}^n \frac{1}{\tilde{\lambda}_i^2} \left<\bm{a}_i,\bm{x}\right>^2 \leq O(n^9)$ and hence $\bm{x} \in O(n^{4.5}) \cdot \tilde{\pazocal{E}}$. Then $T^{-1}(\bm{x}) \in (O(n^5) \cdot K) \cap (\setZ^n \setminus \{\bm{0}\})$. 
\end{proof}
\bibliographystyle{alpha}
\bibliography{main}

\begin{thebibliography}{LLJS90}

\bibitem[Ajt96]{ajtai1996generating}
M.~Ajtai.
\newblock Generating hard instances of lattice problems.
\newblock In {\em Proceedings of the 28th STOC}, pages 99--108. ACM, 1996.

\bibitem[AR05]{LatticeProblemsInNPintersec-coNP-AharonovRegevJACM05}
D.~Aharonov and O.~Regev.
\newblock Lattice problems in {NP} cap conp.
\newblock {\em J. {ACM}}, 52(5):749--765, 2005.

\bibitem[Boh96]{bohman1996sum}
T.~Bohman.
\newblock A sum packing problem of erd{\"o}s and the conway-guy sequence.
\newblock {\em Proceedings of the AMS}, 124(12):3627--3636, 1996.

\bibitem[GJ97]{GJNPCompleteness}
M.~R. Garey and D.~S. Johnson.
\newblock {\em Computers and Intractability: A Guide to the Theory of
  NP-Completeness}.
\newblock W.H. Freeman, 1997.

\bibitem[HR07]{SVPhardness-RegevHavivSTOC07}
I.~Haviv and O.~Regev.
\newblock Tensor-based hardness of the shortest vector problem to within almost
  polynomial factors.
\newblock pages 469--477, 2007.

\bibitem[Joh48]{JohnsTheorem1948}
F.~John.
\newblock Extremum problems with inequalities as subsidiary conditions.
\newblock In {\em Studies and {E}ssays {P}resented to {R}. {C}ourant on his
  60th {B}irthday, {J}anuary 8, 1948}, pages 187--204. Interscience Publishers,
  Inc., New York, N. Y., 1948.

\bibitem[KK82]{KKDifferencing}
N.~Karmarkar and R.~Karp.
\newblock The differencing method of set partitioning.
\newblock Technical report, CS Division, UC Berkeley, 1982.
\newblock
  http://digitalassets.lib.berkeley.edu/techreports/ucb/text/CSD-83-113.pdf.

\bibitem[KPP04]{KnapsackProblemsKellererEtAlBook2004}
H.~Kellerer, U.~Pferschy, and D.~Pisinger.
\newblock {\em Knapsack problems}.
\newblock Springer, 2004.

\bibitem[LLJS90]{lagarias1990korkin}
J.~Lagarias, H.~Lenstra~Jr, and C.~Schnorr.
\newblock Korkin-zolotarev bases and successive minima of a lattice and its
  reciprocal lattice.
\newblock {\em Combinatorica}, 10(4):333--348, 1990.

\bibitem[LLL82]{lenstra1982factoring}
A.~Lenstra, H.~Lenstra, and L.~Lov{\'a}sz.
\newblock Factoring polynomials with rational coefficients.
\newblock {\em Mathematische Annalen}, 261(4):515--534, 1982.

\bibitem[Lov90]{lovasz1990geometric}
L.~Lov{\'a}sz.
\newblock {\em Geometric algorithms and algorithmic geometry}.
\newblock American Mathematical Society, 1990.

\bibitem[Lun88]{lunnon1988integer}
W.~Lunnon.
\newblock Integer sets with distinct subset-sums.
\newblock {\em Mathematics of Computation}, 50(181):297--320, 1988.

\bibitem[LY11]{lev2011size}
V.~Lev and R.~Yuster.
\newblock On the size of dissociated bases.
\newblock {\em the electronic journal of combinatorics}, 18(1):P117, 2011.

\bibitem[Mat02]{LecturesOnDiscreteGeometryMatousek2002}
J.~Matousek.
\newblock {\em Lectures on Discrete Geometry}.
\newblock Springer-Verlag New York, Inc., Secaucus, NJ, USA, 2002.

\bibitem[Mer06]{mertens2006easiest}
S.~Mertens.
\newblock The easiest hard problem: Number partitioning.
\newblock {\em Computational Complexity and Statistical Physics},
  125(2):125--139, 2006.

\bibitem[MG12]{grotschel2012geometric}
A.~Schrijver M.~Gr{\"o}tschel, L.~Lov{\'a}sz.
\newblock {\em Geometric algorithms and combinatorial optimization}, volume~2,
  pages 122--125.
\newblock Springer, 2012.

\bibitem[MT90]{KnapsackProblemsMartelloTothBook1990}
S.~Martello and P.~Toth.
\newblock {\em Knapsack Problems: Algorithms and Computer Implementations}.
\newblock John Wiley \& Sons, Inc., New York, NY, USA, 1990.

\bibitem[NV10]{nguyen2010lll}
P.~Nguyen and B.~Vall{\'e}e.
\newblock The lll algorithm.
\newblock {\em Information Security and}, 2010.

\bibitem[Sch87]{BlockReductionSchnorr87}
C.~Schnorr.
\newblock A hierarchy of polynomial time lattice basis reduction algorithms.
\newblock {\em Theor. Comput. Sci.}, 53:201--224, 1987.

\bibitem[WY92]{woeginger1992equal}
G.~Woeginger and Z.~Yu.
\newblock On the equal-subset-sum problem.
\newblock {\em Information Processing Letters}, 42(6):299--302, 1992.

\end{thebibliography}

\end{document}